\documentclass[9pt,conference]{IEEEtran}
\IEEEoverridecommandlockouts
\usepackage{cite}
\usepackage{amsmath,amssymb,amsfonts}
\usepackage{textcomp}
\usepackage{xcolor}
\usepackage{marvosym}
\usepackage{CJKutf8}
\usepackage{changes}

\usepackage{gensymb}
\usepackage{bm}

\usepackage{amsmath,amssymb}
\usepackage{amsthm}

\usepackage{balance}

\usepackage[figuresright]{rotating}

\usepackage{graphicx}
\graphicspath{{figures/}}
\usepackage[export]{adjustbox}
\usepackage{flushend}

\usepackage{subfigure}
\usepackage{float}

\usepackage{listings}
\usepackage{color}
\usepackage{soul}
\usepackage{tikz}
\newcommand*{\circled}[1]{\lower.7ex\hbox{\tikz\draw (0pt, 0pt)%
    circle (.5em) node {\makebox[1em][c]{\small #1}};}}

\usepackage{algorithm}
\usepackage{algpseudocode}

\theoremstyle{definition}

\newtheorem{theorem}{Theorem}

\newtheorem*{proof}{Proof}
 
\theoremstyle{remark}

\definecolor{mygreen}{rgb}{0,0.6,0}
\definecolor{mygray}{rgb}{0.5,0.5,0.5}
\definecolor{mymauve}{rgb}{0.58,0,0.82}

\usepackage[numbers,sort&compress]{natbib}
\usepackage{xspace}  %
\newcommand{\eat}[1]{}

\newcommand{\stitle}[1]{\vspace{0.5ex}\noindent{\bf #1}}

\newcommand{\bi}{\begin{itemize}}
\newcommand{\ei}{\end{itemize}} 

\newcommand{\ie}{\emph{i.e.,}\xspace}

\newcommand{\wrt}{\emph{w.r.t.}\xspace}

\newcommand{\rf}{\textit{refactor}\xspace}
\usepackage{booktabs} 
\usepackage{multirow}
\usepackage{array}
\newcolumntype{I}{!{\vrule width 1.2pt}}
\newlength\savedwidth

\newlength\savewidth

\usepackage{makecell}

\usepackage{etoolbox}
\makeatletter
\patchcmd{\@makecaption}
  {\scshape}
  {}
  {}
  {}
\makeatletter
\patchcmd{\@makecaption}
  {\\}
  {.\ }
  {}
  {}
\makeatother

\usepackage{geometry}
\begin{document}
\begin{CJK*}{UTF8}{gbsn}


\title{Parallel AIG Refactoring via Conflict Breaking}

\author{\IEEEauthorblockN{
Ye~Cai$^{1}$
Zonglin Yang$^{1,2}$,
Liwei~Ni$^{2,4}$,
Junfeng~Liu$^{3}$,
Biwei~Xie$^{4,2}$,
and Xingquan~Li$^{2,\text{\Letter}}$
}
\IEEEauthorblockA{
$^1$Shenzhen University, Shenzhen, China}
\IEEEauthorblockA{
$^2$Peng Cheng Laboratory, Shenzhen, China}
\IEEEauthorblockA{
$^3$SKLSDE, Beihang University, Beijing, China}
\IEEEauthorblockA{
$^4$Institute of Computing Technology, Chinese Academy of Sciences, Beijing, China}
\IEEEauthorblockA{
Emails: $^{\IEEEauthorrefmark{1}}$caiye@szu.edu.cn, 2100271085@email.szu.edu, nlwmode@gmail.com, $^{\text{\Letter}}$lixq01@pcl.ac.cn }
}
\maketitle

\begin{abstract}

Algorithm parallelization to leverage multi-core platforms for improving the efficiency of Electronic Design Automation~(EDA) tools plays a significant role in enhancing the scalability of Integrated Circuit (IC) designs.
Logic optimization is a key process in the EDA design flow to reduce the area and depth of the circuit graph by finding logically equivalent graphs for substitution, which is typically time-consuming.
To address these challenges, in this paper, we first analyze two types of conflicts that need to be handled in the parallelization framework of refactoring And-Inverter Graph~(AIG). 
We then present a fine-grained parallel AIG refactoring method, which strikes a balance between the degree of parallelism and the conflicts encountered during the refactoring operations.
Experiment results show that our parallel refactor is 28x averagely faster than the sequential algorithm on large benchmark tests with 64 physical CPU cores, and has comparable optimization quality.
\end{abstract}

\begin{IEEEkeywords}
Logic Optimization, Parallelization, Refactoring, AIG
\end{IEEEkeywords}

\section{Introduction}
\label{sec:intro}

With the continuous scaling of technology nodes, the number of transistors that can be integrated into a chip has grown exponentially, and thus the design complexity of Very-Large-Scale Integration (VLSI) implementation is expected to exceed soon the scaling capabilities of many existing Computer-Aided Design (CAD) tools~\cite{CAD, EDA}.
In order to accommodate the continually escalating complexity of Integrated Circuit~(IC) designs on shorter design cycles, design space exploration, and optimization, efficient Electronic Design Automation (EDA) algorithms are essentially needed, which have attracted significant attention in IC design within both the academic and industrial communities~\cite{EDA, CAD, berkeley-abc, mishchenko2011scalable, de2010chip, amaru2017logic, HEBO, qin2018solution}. 

Logic optimization in the EDA design flow refers to the process of improving or transforming a given digital design to meet specific objectives, such as Performance, Power, and Area (PPA), without altering the intended functionality of the design. Due to the vast scale of circuit designs, the designed functionality is typically represented by the logic graph And-Inverter Graph (AIG), which is modeled using \textit{AND2} and \textit{inverter} gates. The \rf~\cite{refactor, adaptive_refactor} is a sequential local replacement operator targeting logic optimization from different granularities of graph or optimization objectives. As shown in Fig. \ref{fig:prelim-aig}, \rf greedily operates on subgraphs (logic cone) to reduce area and depth. It simplifies the logic cone based on Irredundant Sum-Of-Product (ISOP) techniques of Boolean algebra and replaces it with a logic equivalent subgraph constructed from the simplified Boolean function. This sequential local replacement operator is executed multiple times iteratively in practical applications to obtain better PPA that meets circuit design requirements~\cite{HEBO, DSE_ML, DSE_RL, DSE_DL}.

Therefore, in this paper, we introduce a new parallel acceleration algorithm for \rf operation based on AIG, which achieves high acceleration over sequential \rf with comparable Quality of Results (QoRs).  The main contributions are summarized as follows: 

\begin{enumerate}
    \item  Based on our analysis of two types of conflicts, we introduce a levelization approach on AIG, aimed at enhancing the parallelism of the operator  \rf. 
    \item  To achieve a high-performance procedure, we further propose a parallel scheme for \rf， which includes four stages (i) work scheduler, (ii) evaluation, (iii) replacement, and (iv) post-processing to implement the parallelism. 
    \item  We conduct extensive experimental studies on large public circuits, which achieve an average speedup of $28\times$  compared with the sequential \rf in ABC~\cite{berkeley-abc} under $64$ threads, without QoR degradation.
\end{enumerate}

The remainder of this paper is structured as follows. 
In Section \ref{sec:prelim}, we introduce the relevant terms used in this paper. 
The key issues for the designed parallel \rf are analyzed in Section \ref{sec:analyses}. 
In Section \ref{sec:method}, we detail our proposed parallel \rf method, based on the analyses. 
The extensive experimental results and analyses are presented in Section \ref{sec:experiment}, followed by the conclusion in Section \ref{sec:conclusion}.

\section{Preliminaries}
\label{sec:prelim}
\subsection{Background}

\begin{figure}[t]
\centering
\subfigure[Original AIG of circuit function.]{
    \label{fig: aig-a}
    \includegraphics[width=0.45\linewidth]{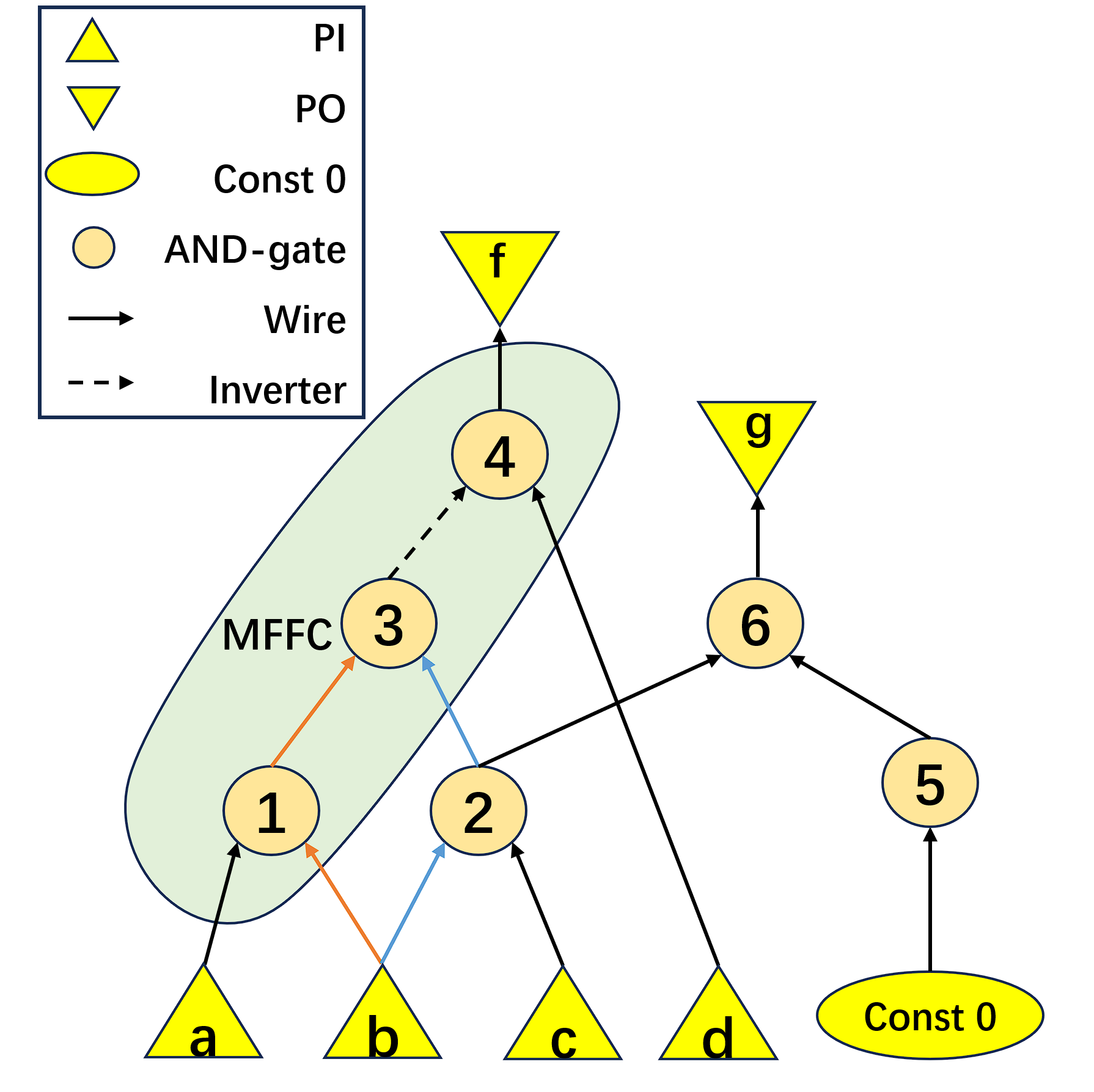}
    }
\subfigure[Optimized AIG by \rf.]{
    \label{fig: aig-b}
    \includegraphics[width=0.45\linewidth]{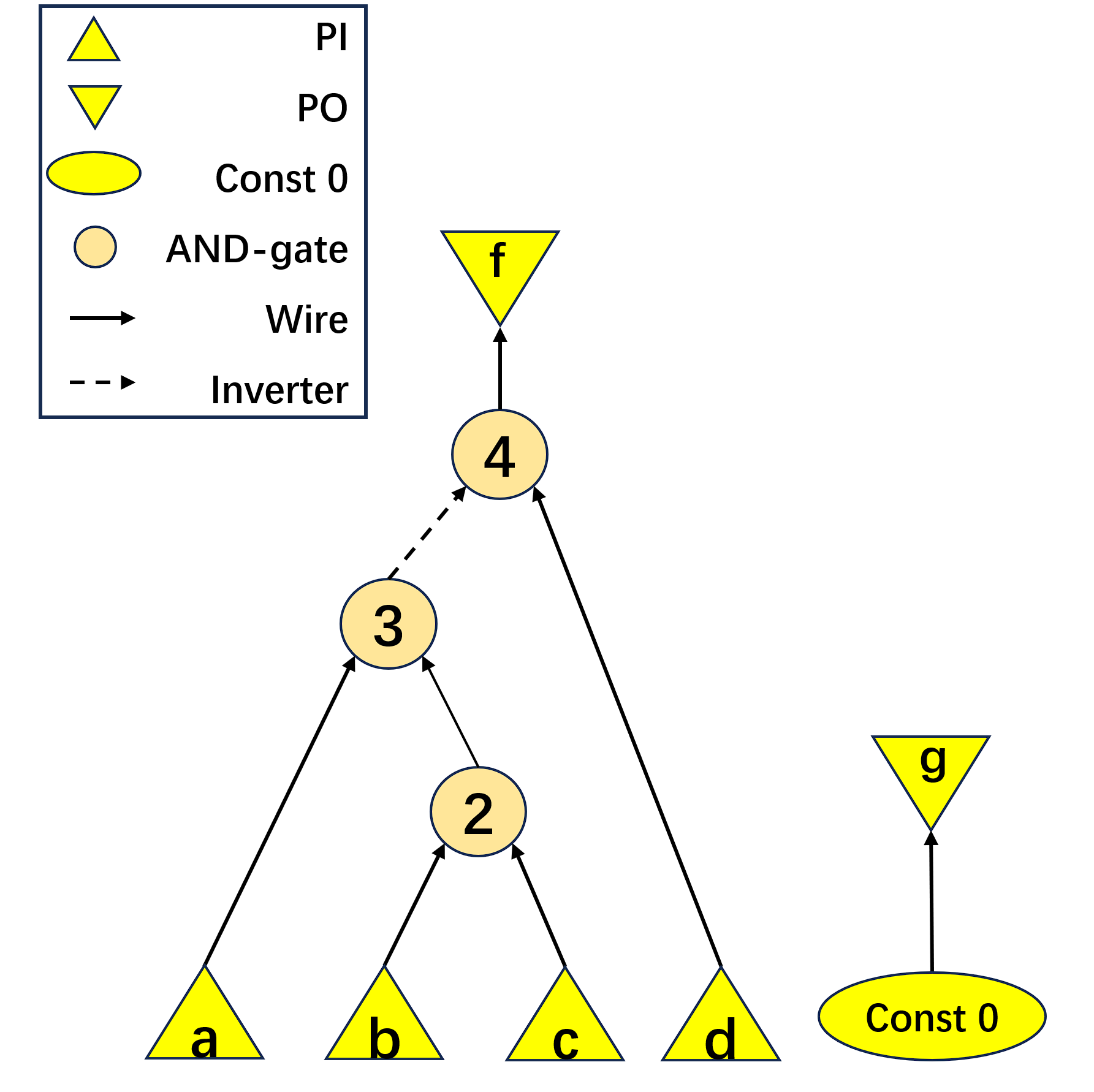}
    }  
\caption{An example of an AIG (left) and its optimised AIG after \rf (right), where the functions   $f=\neg (a \wedge b \wedge c) \wedge d$, and $g=0$. The paths from node $b$ to node $3$ are a pair of reconvergence paths.}
\label{fig:prelim-aig}
\end{figure}

\textbf{And-Inverter Graph (AIG).} It is a typical representation of the logical functionality of a circuit. In AIG, nodes are implemented using \textit{AND2} gates, and edges represent inverted or non-inverted signals, forming a Directed Acyclic Graph (DAG).

\textbf{Cut.} The cut $C$ of a node $v$ is a set of nodes such that every path from Primary Inputs (PIs) to $v$ must pass through at least one node of cut $C$, and \textit{k-cut} refers to the cut contains at most $k$ nodes.
The \textit{logic cone} of a node $v$  is always associated with the cut, which is a set of nodes reachable from the cut leading to node $v$.
The local replacement of \rf is performed on the logic cone as it can be utilized to calculate the Boolean function rooted at $v$. 

\textbf{Maximum Fanout-Free Cone (MFFC).} 
An MFFC of a node $v$ in the AIG is a subset of nodes within its logic cone. In this subset, none of the nodes have fanouts that extend beyond the boundaries of the MFFC.
The size of the MFFC of node $v$ represents the cost required to construct the function of node $v$.

\textbf{Reconvergence.} 
Reconvergence happens when paths from the output of a node intersect before reaching the Primary Outputs (POs)~\cite{refactor,riener2022boolean}. In multi-level logic networks, reconvergence is inevitable due to logic sharing. However, excessive reconvergence often leads to redundancy, indicating the potential for logic optimization~\cite{refactor}.

\textbf{Structural Hashing.} 
Structural hashing is used to determine the existence of a functional node. Its principle is a non-redundant hash table, where nodes with the same function are mapped to the same key corresponding to a linked list.

\section{Analyses of Parallel {Refactor} Design}
\label{sec:analyses}

In this section, we analyze key design issues of parallel \rf solutions on the AIG. The serial \rf algorithm processes the corresponding subtask(\emph{node-refactoring}) of each node in topological order, calculating a more optimized subgraph for replacement. Our goal is to execute these subtasks in parallel and resolve the conflicts that arise. Through analysis, we have identified two types of conflicts:

\begin{figure}[t]
\centering
\subfigure[Origin AIG.]{
    \label{fig:xyz}
    \includegraphics[width=0.45\linewidth]{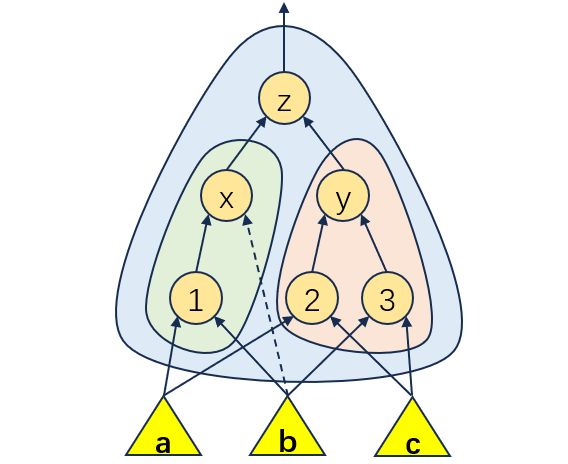}
    }
\hspace{-2ex}
\subfigure[AIG after refactor  $x$ by thread~$t_1$.]{
    \label{fig:x}
    \includegraphics[width=0.45\linewidth]{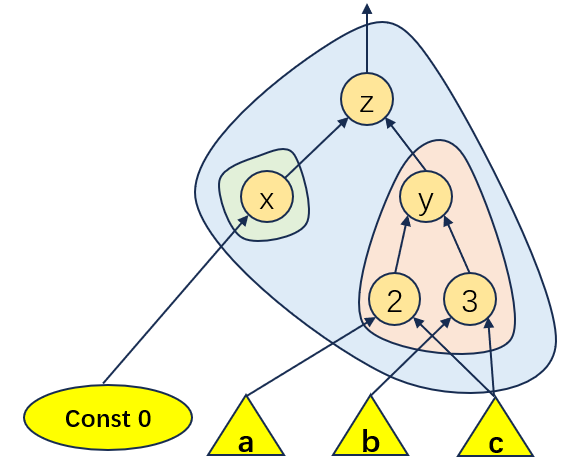}
    }
    
\subfigure[AIG after refactor  $y$ by thread~$t_2$.]{
    \label{fig:y}
    \includegraphics[width=0.45\linewidth]{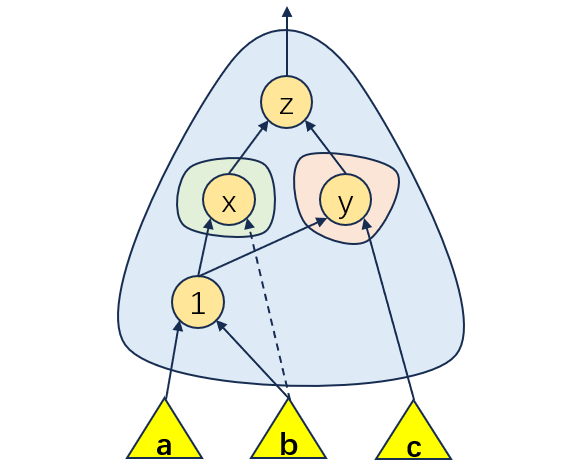}
    }
\hspace{-1ex}    
\subfigure[AIG after refactor $z$ by thread~$t_3$.]{
    \label{fig:z}
    \includegraphics[width=0.45\linewidth]{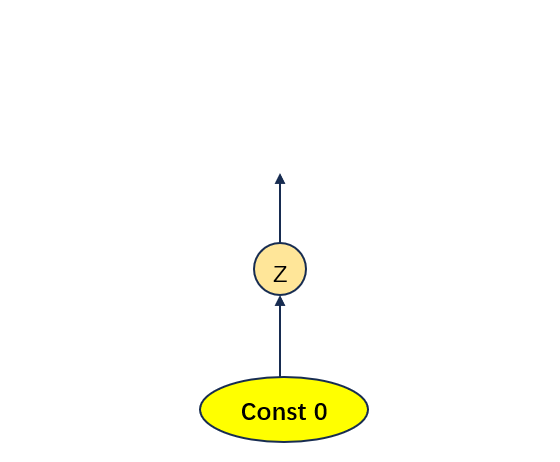}
    }  
\caption{Example of the conflicts caused by data dependency.}
\label{fig:preliminary-conflicts}
\end{figure}

\stitle{Conflict Type I: Data Race.} {\em Different threads inevitably modify the shared data of (1) the original AIG and (2) its structure hashing table, which results in conflicts.
\vspace{.5ex} } 

(1) The data structure of AIG contains an array of pointers to nodes, conflicts could obviously occur when multiple threads modify the node array or the data of the same node.

(2) The structure hashing table is usually organized by an array of the linked list.
Thus, conflicts occur when multiple threads modify the same linked list.

Note that, the other two shared data, \ie cuts and MFFCs are stored as temporary variables, such that they are independent of the operations of different threads to avoid conflicts. 

\stitle{Conflict Type II: Data Dependency.} { \em  Due to the fine-grained node-refactoring operations in MFFCs, the modification on MFFCs could be dependent on other threads, leading to conflicts of (1) duplicate node deletion, and (2) concurrent modification-deletion.}
\vspace{.5ex}

In fact, only the nodes in the MFFC should be modified during the \textit{replacement} stage of the \textit{node-refactoring} process. 
Therefore, based on the  dependencies of MFFCs \wrt different nodes in multiple threads, we conduct the conflicts

(1) During \textit{replacement}, the steps of \rf are to first delete the nodes in the MFFC and then create a new subgraph. 
When multiple threads both delete the nodes on the overlapped MFFCs, the conflict of duplicate node deletion arises.
Taking Fig. \ref{fig:preliminary-conflicts} as an example, if threads $t_1$ and $t_3$ simultaneously perform refactoring on nodes $x$ and $z$ respectively, they might duplicately delete node $1$, as the MFFCs of $x$  and $z$ are overlapped, shown in Figs. \ref{fig:x} \& \ref{fig:z}.

(2) Concurrent modification-deletion conflicts arise when multiple threads attempt to delete and reuse the same node simultaneously. 
For example, considering Fig. \ref{fig:preliminary-conflicts}, if thread $t_1$ and thread $t_2$ perform node-refactoring on nodes $x$ and $y$, respectively, 
at the same time, thread $t_1$ needs to delete node 1 while thread $t_2$ needs to reuse node 1. This results in a conflict, as shown in Fig. \ref{fig:x} and \ref{fig:y}.

Note that, two disconnected subgraphs in AIG do not have same PI, so that the logic functions corresponding to their root nodes are completely unrelated, and the above two conflicts should not occur.

\section{method}
\label{sec:method}

\begin{figure}[t]
    \centering
    \includegraphics[width=1.0\linewidth]{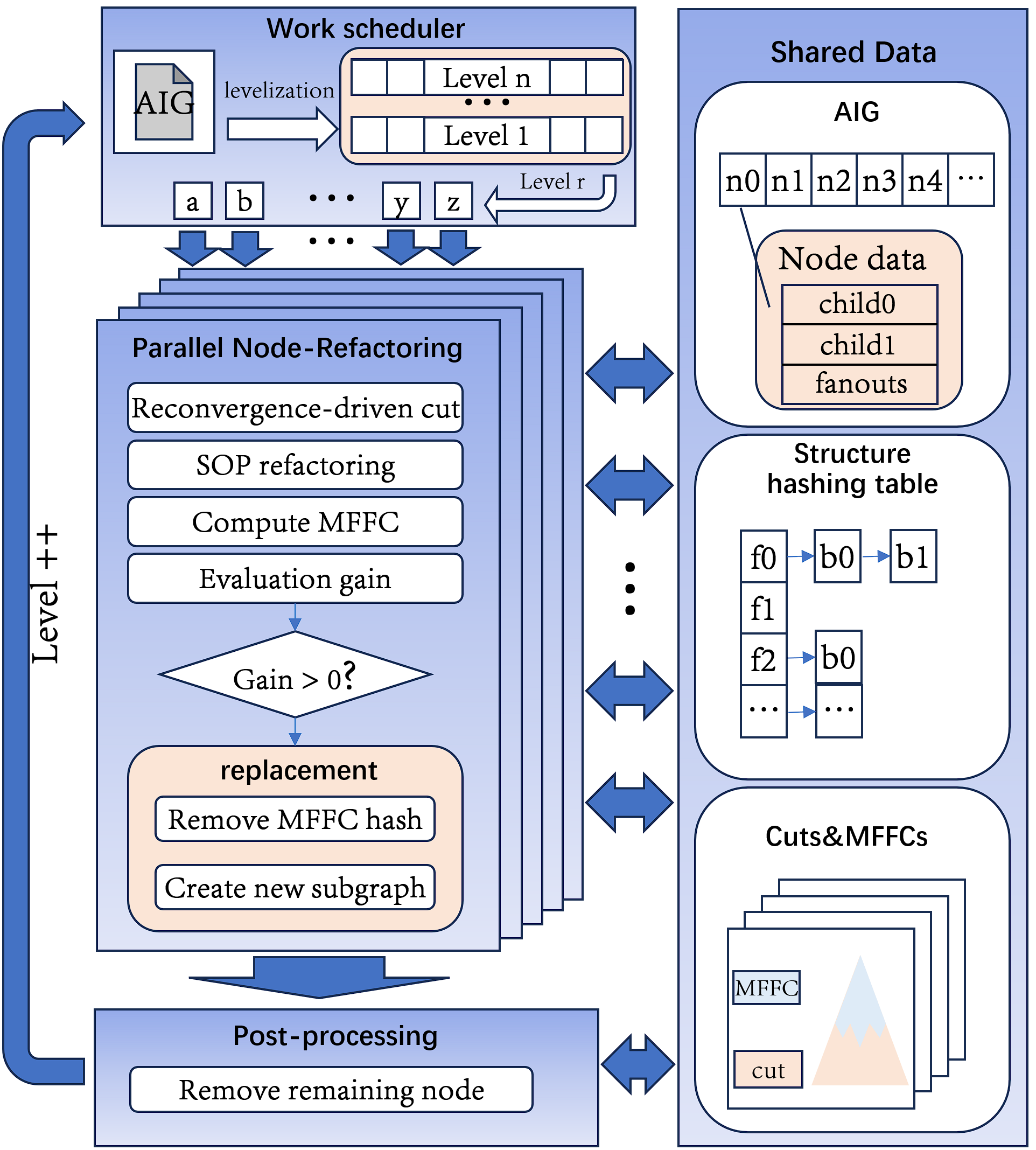}
    \caption{The framework of the proposed parallel AIG refactoring.}
    \label{fig:method-flow}
\end{figure}
Our proposed parallel refactoring flow is shown in Fig. \ref{fig:method-flow}, it consists of iterative four parts: (1) work scheduler; (2) parallel node-refactoring; (3) replacement; (4) post-processing.

\subsection{Work Scheduler}
\label{work scheduler}
In order to solve conflict II.(1) \textit{duplicate node deletion}, 
we parallel process the nodes that do not overlap with MFFCs, which is inspired by the properties of MFFCs ensured by following Theorems.

\begin{theorem}
\label{theorem1}
\em 
If one node is contained in both MFFCs of node $x$ and node $y$, there is a path between $x$ and  $y$. (The path here only allows one-way transmission from PIs to POs). \hfill $\square$
\end{theorem}
\begin{proof}
The fanouts of all MFFC nodes, except the root node, are contained within the MFFC itself. They can only connect to nodes outside the MFFC via the root node. Suppose node $n$ belongs to the MFFCs of both node $x$ and node $y$, and there are paths from node $n$ to the root nodes $x$ and $y$.
(1) If node $x$ is also in the MFFC of node $y$, there is a path from node $x$ to the root node $y$.
(2) If node $x$ is not in the MFFC of node $y$, then there is a path from node $n$ to node $x$. Since any path from node $n$ to nodes outside the MFFC of node $y$ has to pass through the root node $y$, i.e., a path $n \rightarrow y \rightarrow x$ exists.
Combining these, we have proved Theorem \ref{theorem1}.
\hfill $\square$
\end{proof}

\begin{theorem}
\label{theorem2}
\em
Given a set of nodes that have the same level in an AIG, the MFFCs of these nodes do not overlap with each other. \hfill $\square$
\end{theorem}

\begin{proof}
By Theorem \ref{theorem1}, if the MFFCs of two root nodes $x$ and $y$ have overlap, that is, there exists at least one overlapped node $n$, then there is a path between node $x$ and node $y$. 
However, all the nodes on the path have different levels, and hence there is no path between two nodes with the same level, nor is there any MFFC overlap. 
\hfill $\square$
\end{proof}

\stitle{Addressing conflict II.(1) duplicate node deletion.}
Theorems \ref{theorem1} \& \ref{theorem2} tell us the MFFCs of nodes with the same level do not overlap with each other. 
Therefore, to address conflict II.(1)  duplicate node deletion as mentioned before, the \textit{work scheduler} organizes nodes in the AIG by levelization and divides nodes into different groups for parallel processing. The scheduler processes each group of nodes in order from PIs to POs direction, and all threads can only process node-refactoring for nodes from the same group simultaneously. This scheduling method ensures that the MFFCs of nodes in the same group do not overlap and that the global processing order of nodes follows the topological order.

\subsection{Evaluation}
\label{evaluation}
Similar to the sequential \rf, the purpose of the evaluation is to determine whether the current logic cone structure should be replaced with the new subgraph created through SOP \textit{refactoring}, by evaluating their gains.
Here, in parallel refactor design, we additionally address the conflict involved in data race and data dependency: I.(1) the original AIG and II.(2) concurrent modification-deletion.

\stitle{Addressing conflict I.(1) the original AIG.}
The cuts and MFFCs, are stored as global variables in the AIG data structure, which leads to conflict. 
To prevent this, we allocate the memories of these global variables to the local variables for each thread, eliminating the need for communication in the parallel environment.

\stitle{Addressing conflict II.(2) concurrent modification-deletion.}
As explained in Section \ref{sec:analyses}, conflict II.(2) can occur when multiple threads concurrently attempt to delete and reuse the same node, as depicted in Fig. \ref{fig:preliminary-conflicts}. 
We use flags to mark the reused nodes in the MFFCs, preventing their deletion during subsequent MFFC deletions and thus avoiding conflicts.

\subsection{Replacement}
The \textit{replacement} stage mainly contains node deletion and creation of nodes,  which involves data race for I.(1) the original AIG and I.(2) the structure hashing table. 
Besides,  conflict II.(2) concurrent modification-deletion is further clarified.

\label{replacement}
\begin{figure}[t]
    \centering
    \includegraphics[width=1.0\linewidth]{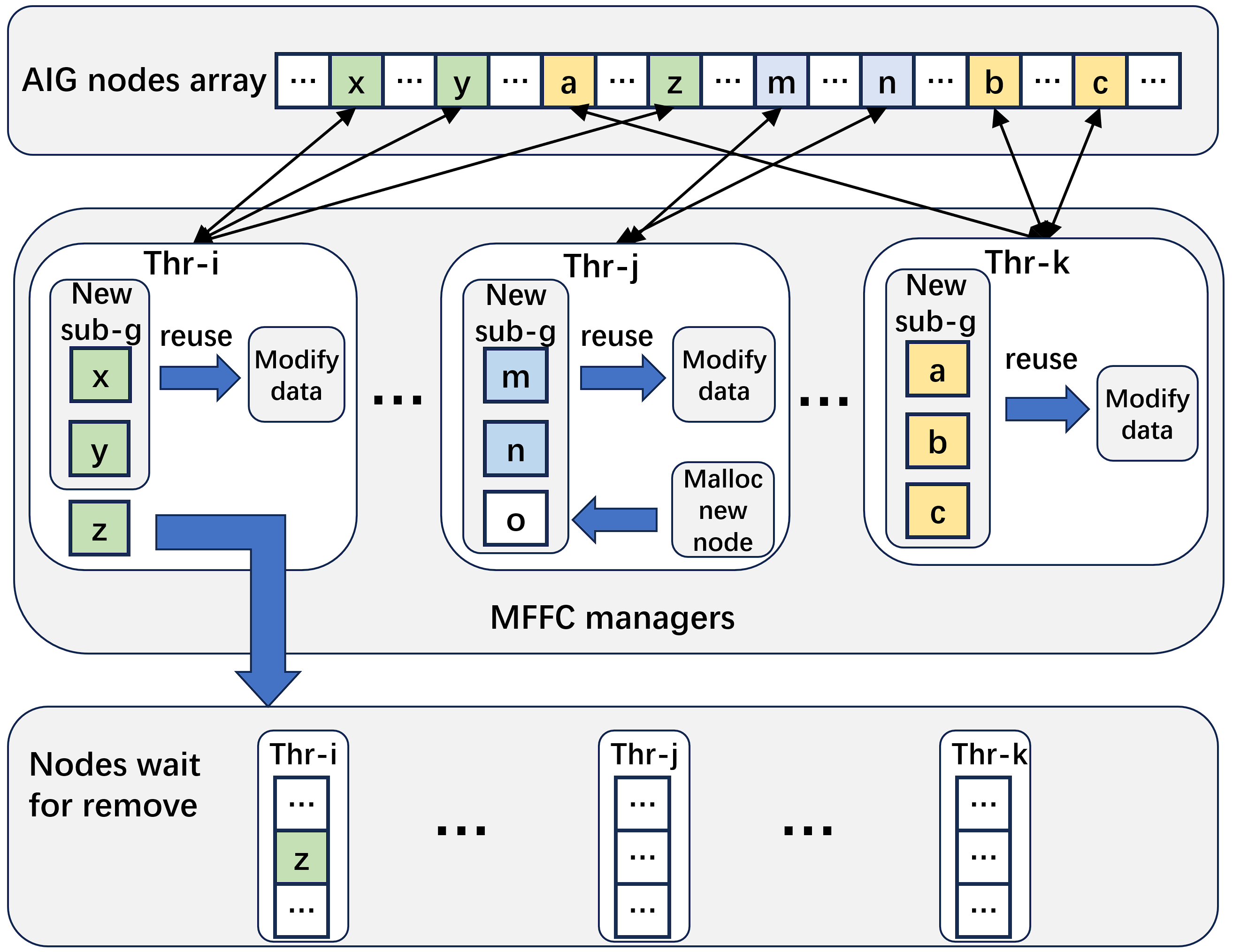}
    \caption{Example of MFFC node recycling and deletion.}
    \label{fig:method-recycle}
\end{figure}

\stitle{Addressing conflict Data Race I.(1) and I.(2).}
The first step of \textit{replacement} is to delete the nodes in MFFC, and the deletion of nodes requires modifications on both the AIG and the structure hashing table. 
(1) Locking on the global AIG and the structure hashing table could significantly impact the thread blocking.
Due to the data structure organization, the blocking caused by locking the structure hashing table(array of linked lists) is very small. 
Because of the design of the hash function, nodes are rarely mapped to the same linked list in practice.
(2) However, the blocking caused by locking the AIG nodes array is unacceptable.
Therefore, each thread collects the node IDs to be deleted and performs the deletion in the AIG in the after \textit{post-proccessing} phase, which processes sequentially.

To further reduce the modification of AIG, we propose the MFFC node recycling strategy. That is, we directly change the data of the MFFC nodes to be deleted to the data of the nodes in the new subgraph. The reasons for doing this are as follows. (1) In the traditional process, the nodes in MFFC are deleted and their memory space is recycled. When constructing the new subgraph, new memory space is allocated for the nodes and their information is filled in. However, the operation of memory space is quite time-consuming. (2) The cost of the new subgraph is always lower than that of the old subgraph, so the number of nodes in MFFC is always enough. (3) This reduces the operation of AIG and avoids some conflicts.

As depicted in Fig. \ref{fig:method-recycle}, the MFFC node recycling strategy can lead to three possible scenarios: (1) If thread $i$ has surplus nodes available for recycling, nodes $x$ and $y$ are reused, while node $z$ is added to the deletion queue. (2) If thread $j$ lacks sufficient nodes for recycling, it must lock the nodes array and allocate new memory space. However, this situation did not occur during our experimental evaluation. (3) If thread $k$ has just enough nodes for recycling, it can directly modify their information. All the nodes in the deletion queue will be deleted in the sequential post-processing phase.

\section{experiment}
\label{sec:experiment}
\subsection{Setup}

\stitle{Environment.}
We implemented the proposed method using GCC 11.1.0 x86 version on a 64-bit Ubuntu 18.04.6 system. We collected our experimental results on a system with 128G shared RAM and 2 Intel® Xeon® Gold 6252 CPU @ 2.10GHz processors, each with 48 physical cores.

\stitle{Benchmarks.}
Since the \textit{parallel refactor} is specifically designed for large AIGs, we expanded the largest 10 circuits from the EPFL~\cite{benchmark_epfl} benchmark suite to the tens of millions of gates level using the \texttt{double} command in ABC~\cite{berkeley-abc}. These expanded circuits were used as input for our experimental analysis. The details of the expanded dataset are shown in Table \ref{table: benchmark details}.

\begin{table}[tb!]
	\centering
	\caption{Benchmark details.}
	\label{table: benchmark details}
    \renewcommand\tabcolsep{3.8pt} 
		\begin{tabular}{|c|c|rrrr|}
			\hline
			\multirow{1}{*}{}        & \multirow{1}{*}{Benchmark name} & \multicolumn{1}{r|}{PIs}                   & \multicolumn{1}{r|}{POs}      & \multicolumn{1}{r|}{Area}       & {Depth} \\ \hline
			\multirow{7}{*}{arith}   & log2\_10xd                      & \multicolumn{1}{r|}{\raggedleft{32,768}  } & \multicolumn{1}{r|}{32,768}   & \multicolumn{1}{r|}{32,829,440} & 444     \\
			                         & square\_10xd                    & \multicolumn{1}{r|}{\raggedleft{65,536}  } & \multicolumn{1}{r|}{131,072}  & \multicolumn{1}{r|}{18,927,616} & 250     \\
			                         & sin\_10xd                       & \multicolumn{1}{r|}{\raggedleft{24,576}  } & \multicolumn{1}{r|}{25,600}   & \multicolumn{1}{r|}{5,545,984}  & 225     \\
			                         & div\_10xd                       & \multicolumn{1}{r|}{\raggedleft{131,072} } & \multicolumn{1}{r|}{131,072}  & \multicolumn{1}{r|}{58,620,928} & 4,372   \\
			                         & hyp\_8xd                        & \multicolumn{1}{r|}{\raggedleft{65,536}  } & \multicolumn{1}{r|}{32,768}   & \multicolumn{1}{r|}{54,869,760} & 24,801  \\
			                         & sqrt\_10xd                      & \multicolumn{1}{r|}{\raggedleft{131,072} } & \multicolumn{1}{r|}{65,536}   & \multicolumn{1}{r|}{25,208,832} & 5,058   \\
			                         & multiplier\_10xd                & \multicolumn{1}{r|}{\raggedleft{131,072} } & \multicolumn{1}{r|}{131,072}  & \multicolumn{1}{r|}{27,711,488} & 274     \\ \hline
			\multirow{3}{*}{control} & arbiter\_10xd                   & \multicolumn{1}{r|}{\raggedleft{262,144} } & \multicolumn{1}{r|}{132,096}  & \multicolumn{1}{r|}{12,123,136} & 87      \\
			                         & voter\_10xd                     & \multicolumn{1}{r|}{\raggedleft{1025,024}} & \multicolumn{1}{r|}{1,024}    & \multicolumn{1}{r|}{14,088,192} & 70      \\
			                         & mem\_ctrl\_10xd                 & \multicolumn{1}{r|}{\raggedleft{1232,896}} & \multicolumn{1}{r|}{1260,544} & \multicolumn{1}{r|}{47,960,064} & 114     \\ \hline
		\end{tabular}
\end{table}

\begin{figure*}[t]
    \centering
    \includegraphics[width=0.7\textwidth]{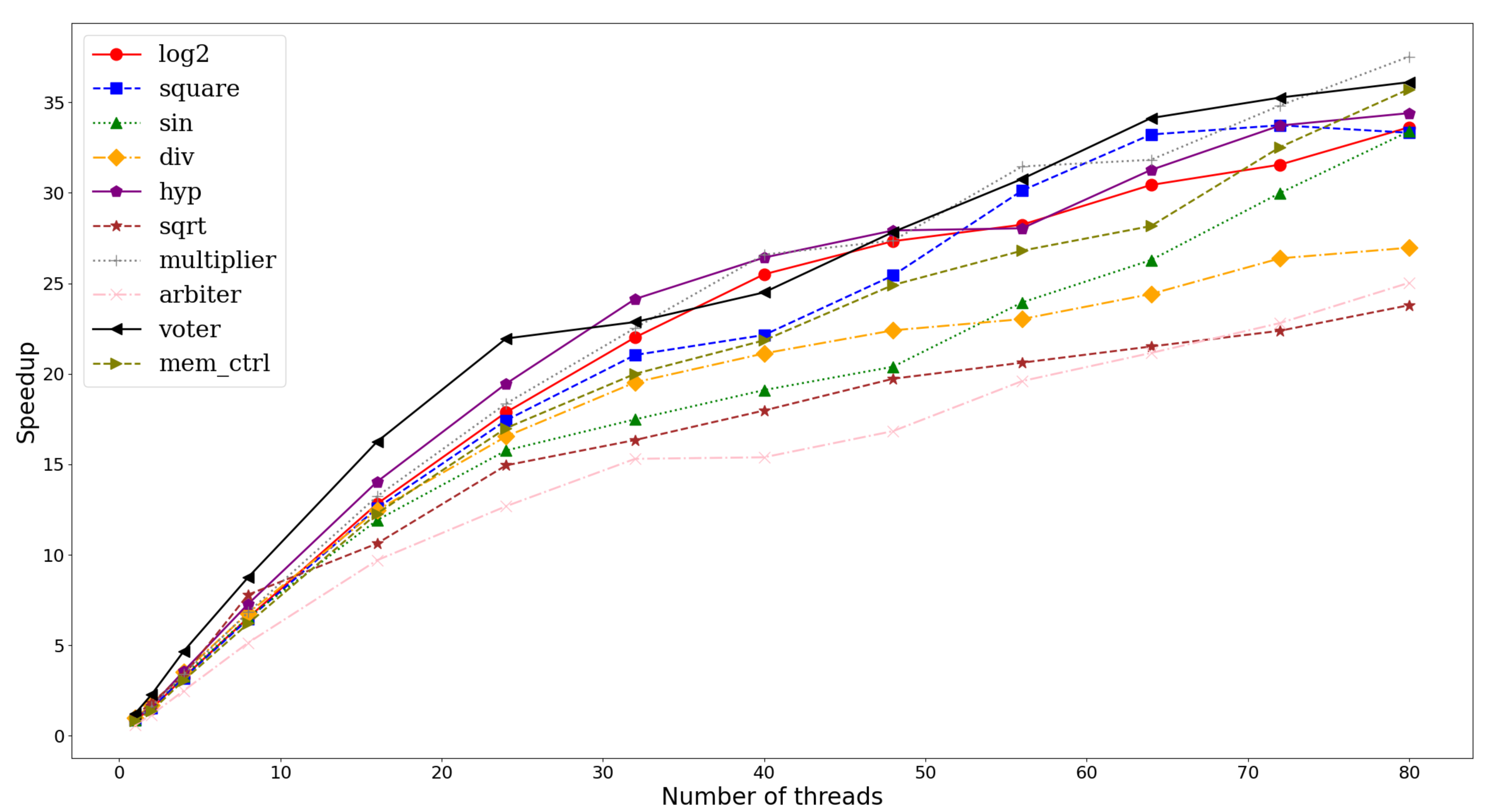}
    \caption{Scalability of \texttt{parallel refactor} compared with ABC command \texttt{refactor -l 0}.}  
    \label{fig: result}
\end{figure*}

\begin{table}[t]
	\centering
	\caption{Runtimes, in seconds, threads in 1, 40, 80 compared with ABC command \texttt{refactor -l 0}, the time unit is second.}
	\label{table: runtimes}
		\begin{tabular}{|c|rrrrr|}
			\hline
			\multirow{1}{*}{Benchmark name} & \multicolumn{1}{r|}{ABC}      & \multicolumn{1}{r|}{1 thr}    & \multicolumn{1}{r|}{40 thr} & \multicolumn{1}{r|}{64 thr} & {80 thr} \\ \hline
			log2\_10xd                      & \multicolumn{1}{r|}{585.44}   & \multicolumn{1}{r|}{658.97}   & \multicolumn{1}{r|}{22.96}  & \multicolumn{1}{r|}{19.24} & 17.41    \\
			square\_10xd                    & \multicolumn{1}{r|}{347.32}   & \multicolumn{1}{r|}{399.65}   & \multicolumn{1}{r|}{15.68}  & \multicolumn{1}{r|}{10.45} & 10.42    \\
			sin\_10xd                       & \multicolumn{1}{r|}{67.66}    & \multicolumn{1}{r|}{77.64}    & \multicolumn{1}{r|}{3.54}   & \multicolumn{1}{r|}{2.57}  & 2.03     \\
			div\_10xd                       & \multicolumn{1}{r|}{697.22}   & \multicolumn{1}{r|}{703.04}   & \multicolumn{1}{r|}{33.00}  & \multicolumn{1}{r|}{28.56} & 25.85    \\
			hyp\_8xd                        & \multicolumn{1}{r|}{1,435.47} & \multicolumn{1}{r|}{1,529.36} & \multicolumn{1}{r|}{54.33}  & \multicolumn{1}{r|}{45.91} & 41.73    \\
			sqrt\_10xd                      & \multicolumn{1}{r|}{99.68}    & \multicolumn{1}{r|}{89.39}    & \multicolumn{1}{r|}{5.55}   & \multicolumn{1}{r|}{4.63}  & 4.19     \\
			multiplier\_10xd                & \multicolumn{1}{r|}{565.71}   & \multicolumn{1}{r|}{639.51}   & \multicolumn{1}{r|}{21.28}  & \multicolumn{1}{r|}{17.78} & 15.08    \\ \hline
			arbiter\_10xd                   & \multicolumn{1}{r|}{92.86}    & \multicolumn{1}{r|}{153.78}   & \multicolumn{1}{r|}{6.03}   & \multicolumn{1}{r|}{4.39}   & 3.71     \\
			voter\_10xd                     & \multicolumn{1}{r|}{155.86}   & \multicolumn{1}{r|}{126.18}   & \multicolumn{1}{r|}{6.36}   & \multicolumn{1}{r|}{4.57}  & 4.32     \\
			mem\_ctrl\_10xd                 & \multicolumn{1}{r|}{298.33}   & \multicolumn{1}{r|}{352.49}   & \multicolumn{1}{r|}{13.65}  & \multicolumn{1}{r|}{10.59} & 8.35     \\ \hline
   \textbf{Average Speedup}                 & \multicolumn{1}{r|}{1}        & \multicolumn{1}{r|}{0.92}     & \multicolumn{1}{r|}{22.05}  & \multicolumn{1}{r|}{\textbf{28.24}}  & 31.98    \\ \hline
		\end{tabular}
\end{table}

\stitle{Parameters.}
Due to the absence of similar comparable works, we only compared the speedup and QoR of our proposed \texttt{parallel refactor} with the sequential refactor in ABC. For each case, we performed 5 runs at various thread counts ranging from 1 to 80 and averaged the results for presentation. This approach was adopted to mitigate the impact of external noise.

\subsection{Parallel AIG Refactor Scalability.}
To evaluate the scalability of \textit{parallel refactor}, we present the speedup performance of our algorithm using different numbers of threads in Fig. \ref{fig: result}. We varied the number of threads from 1 to 80, with increments of 8 after 4 threads. We also present the detailed running time of four evenly spaced thread numbers (1, 40, 64 and 80) in Table \ref{table: runtimes}. The average speedup is computed by the average speedup of each case under the same number of threads.

As the number of threads increases, the speedup ratio and the ratio of threads decrease. To analyze the experimental results, we use a performance analysis tool. We find that the bottleneck of parallel acceleration of our algorithm lies in the sequential execution of the \textit{post-processing} stage.

\subsection{Parallel AIG Refactor QoR}

\begin{table}[t]
	\centering
	\caption{QoR of parallel refactor.}
    \renewcommand\tabcolsep{9.4pt} 
	\label{table: QoR}
		\begin{tabular}{|c|rr|rr|}
			\hline
			  \multirow{2}{*}{Benchmark name}  & \multicolumn{2}{c|}{ABC \textit{refactor -l 0}} & \multicolumn{2}{c|}{80 threads/ABC}                 \\ \cline{2-5}
			                                  & \multicolumn{1}{r|}{Area}                       & {Depth}                             & \multicolumn{1}{r|}{Area} & {Depth} \\ \hline
			log2\_10xd                      & \multicolumn{1}{r|}{32,273,408}                 & \multicolumn{1}{r|}{445}            & \multicolumn{1}{r|}{1.00} & 1.00    \\
			square\_10xd                    & \multicolumn{1}{r|}{18,741,248}                 & \multicolumn{1}{r|}{250}            & \multicolumn{1}{r|}{1.00} & 1.00    \\
			sin\_10xd                       & \multicolumn{1}{r|}{5,448,704}                  & \multicolumn{1}{r|}{224}            & \multicolumn{1}{r|}{1.00} & 1.00    \\
			div\_10xd                       & \multicolumn{1}{r|}{58,106,880}                 & \multicolumn{1}{r|}{4,372}          & \multicolumn{1}{r|}{\textbf{0.87}} & 1.00    \\
			hyp\_8xd                        & \multicolumn{1}{r|}{54,359,296}                 & \multicolumn{1}{r|}{24,801}         & \multicolumn{1}{r|}{\textbf{0.99}} & 1.00    \\
			sqrt\_10xd                      & \multicolumn{1}{r|}{23,358,464}                 & \multicolumn{1}{r|}{5,932}          & \multicolumn{1}{r|}{1.00} & 1.00    \\
			multiplier\_10xd                & \multicolumn{1}{r|}{27,457,536}                 & \multicolumn{1}{r|}{274}            & \multicolumn{1}{r|}{1.00} & 1.00    \\ \hline
			arbiter\_10xd                   & \multicolumn{1}{r|}{12,123,136}                 & \multicolumn{1}{r|}{87}             & \multicolumn{1}{r|}{1.00} & 1.00    \\
			voter\_10xd                     & \multicolumn{1}{r|}{12,949,504}                 & \multicolumn{1}{r|}{65}             & \multicolumn{1}{r|}{1.00} & \textbf{1.09}    \\
			mem\_ctrl\_10xd                 & \multicolumn{1}{r|}{47,691,776}                 & \multicolumn{1}{r|}{114}            & \multicolumn{1}{r|}{1.00} & 1.00    \\ \hline
		\end{tabular}
\end{table}

We compare our parallel algorithm with ABC using the results of our algorithm with 80 threads, which is the highest number of threads in our experiments. Table \ref{table: QoR} shows the detailed data of area and depth after ABC command \texttt{refactor -l 0} optimization. We divide the optimization results of our algorithm by those of ABC and keep only two decimal places to show the difference in QoR between the two methods.

As shown in table \ref{table: QoR}, compared to the results in ABC, most examples have nearly identical area and depth. In addition, it is important to note that the variation rate of the QoR results under all the same conditions does not exceed one hundred thousandths. 

\section{conclusion}
\label{sec:conclusion}
In this work, we propose a parallel refactor algorithm for AIG. We identify and address the challenges in parallel scenarios. We adopt a scheduling strategy that alternates between parallel and serial execution to exploit the parallelism among nodes while preserving topological order and optimization quality. For the work scheduler, evaluation, and replacement stages, we apply targeted strategies to exploit the parallelism across node levels while using lightweight locks for accessing shared data. Thus, our algorithm achieves significant speedup while maintaining optimization quality comparable to that of ABC.

\section*{ACKNOWLEDGEMENT}
This work is supported in part by the Major Key Project of PCL (No. PCL2023A03),
the SZU graduate Program SZUGS2023JG06.

\small
\bibliographystyle{IEEEtran} 
\bibliography{formal.bib}

\begin{thebibliography}{10}
\providecommand{\url}[1]{#1}
\csname url@samestyle\endcsname
\providecommand{\newblock}{\relax}
\providecommand{\bibinfo}[2]{#2}
\providecommand{\BIBentrySTDinterwordspacing}{\spaceskip=0pt\relax}
\providecommand{\BIBentryALTinterwordstretchfactor}{4}
\providecommand{\BIBentryALTinterwordspacing}{\spaceskip=\fontdimen2\font plus
\BIBentryALTinterwordstretchfactor\fontdimen3\font minus \fontdimen4\font\relax}
\providecommand{\BIBforeignlanguage}[2]{{%
\expandafter\ifx\csname l@#1\endcsname\relax
\typeout{** WARNING: IEEEtran.bst: No hyphenation pattern has been}%
\typeout{** loaded for the language `#1'. Using the pattern for}%
\typeout{** the default language instead.}%
\else
\language=\csname l@#1\endcsname
\fi
#2}}
\providecommand{\BIBdecl}{\relax}
\BIBdecl

\bibitem{CAD}
T.-W. Huang, ``{A General-purpose Parallel and Heterogeneous Task Programming System for VLSI CAD},'' in \emph{2020 IEEE/ACM International Conference On Computer Aided Design}, 2020, pp. 1--2.

\bibitem{EDA}
L.~Stok, ``Developing parallel eda tools [the last byte],'' \emph{IEEE Design \& Test}, vol.~30, no.~1, pp. 65--66, 2013.

\bibitem{berkeley-abc}
R.~Brayton and A.~Mishchenko, ``{ABC: An academic industrial-strength verification tool},'' in \emph{Proceedings of 22nd International Conference on Computer Aided Verification.}\hskip 1em plus 0.5em minus 0.4em\relax Springer, 2010, pp. 24--40.

\bibitem{mishchenko2011scalable}
A.~Mishchenko, R.~Brayton, J.-H.~R. Jiang, and S.~Jang, ``Scalable don't-care-based logic optimization and resynthesis,'' \emph{ACM Transactions on Reconfigurable Technology and Systems}, vol.~4, no.~4, pp. 1--23, 2011.

\bibitem{de2010chip}
G.~De~Micheli, ``Chip challenge,'' \emph{IEEE Solid-State Circuits Magazine}, vol.~2, no.~4, pp. 22--26, 2010.

\bibitem{amaru2017logic}
L.~Amar{\'u}, P.~Vuillod, J.~Luo, and J.~Olson, ``Logic optimization and synthesis: Trends and directions in industry,'' in \emph{Design, Automation \& Test in Europe Conference \& Exhibition, 2017}.\hskip 1em plus 0.5em minus 0.4em\relax IEEE, 2017, pp. 1303--1305.

\bibitem{HEBO}
A.~Grosnit, C.~Malherbe, R.~Tutunov, X.~Wan, J.~Wang, and H.~B. Ammar, ``Boils: Bayesian optimisation for logic synthesis,'' in \emph{2022 Design, Automation \& Test in Europe Conference \& Exhibition}, 2022, pp. 1193--1196.

\bibitem{qin2018solution}
Y.~Qin, Y.~Ding, Y.~Tan, and Q.~Wu, ``Solution space adjustable cnf obfuscation for privacy-preserving sat solving,'' in \emph{ISPA/IUCC/BDCloud/SocialCom/SustainCom}.\hskip 1em plus 0.5em minus 0.4em\relax IEEE, 2018, pp. 1--8.

\bibitem{refactor}
B.~R. Alan~Mishchenko, ``{Scalable logic synthesis using a simple circuit structure},'' in \emph{Proc. IWLS}, vol.~6, 2006, pp. 15--22.

\bibitem{adaptive_refactor}
L.~Ni, Z.~Yang, J.~Zhang, J.~Liu, H.~Li, B.~Xie, and X.~Li, ``Adaptive reconvergence-driven aig rewriting via strategy learning,'' in \emph{2023 IEEE 41st International Conference on Computer Design (ICCD)}, 2023, pp. 336--343.

\bibitem{DSE_ML}
W.~Haaswijk, E.~Collins, B.~Seguin, M.~Soeken, F.~Kaplan, S.~Süsstrunk, and G.~De~Micheli, ``{Deep Learning for Logic Optimization Algorithms},'' in \emph{2018 IEEE International Symposium on Circuits and Systems}, 2018, pp. 1--4.

\bibitem{DSE_RL}
C.~Wang, C.~Chen, D.~Li, and B.~Wang, ``Rethinking reinforcement learning based logic synthesis,'' \emph{arXiv preprint arXiv:2205.07614}, 2022.

\bibitem{DSE_DL}
A.~Hosny, S.~Hashemi, M.~Shalan, and S.~Reda, ``{DRiLLS: Deep reinforcement learning for logic synthesis},'' in \emph{2020 25th Asia and South Pacific Design Automation Conference}.\hskip 1em plus 0.5em minus 0.4em\relax IEEE, 2020, pp. 581--586.

\bibitem{riener2022boolean}
H.~Riener, S.-Y. Lee, A.~Mishchenko, and G.~De~Micheli, ``Boolean rewriting strikes back: Reconvergence-driven windowing meets resynthesis,'' in \emph{27th IEEE Asia and South Pacific Design Automation Conference}, 2022, pp. 395--402.

\bibitem{benchmark_epfl}
L.~Amar{\'u}, P.-E. Gaillardon, and G.~De~Micheli, ``{The EPFL combinational benchmark suite},'' in \emph{Proceedings of the 24th International Workshop on Logic \& Synthesis}, no. CONF, 2015.

\end{thebibliography}

\end{CJK*}
\end{document}